\documentclass[english]{article}%
\usepackage{amssymb}
\usepackage[T1]{fontenc}
\usepackage[latin9]{inputenc}
\usepackage[letterpaper]{geometry}
\usepackage{amsmath}
\usepackage{setspace}
\usepackage{babel}
\usepackage{apacite}
\usepackage{amsfonts}
\usepackage{graphicx}%
\providecommand{\U}[1]{\protect\rule{.1in}{.1in}}
\newtheorem{theorem}{Theorem}

\newtheorem{definition}{Definition}

\newtheorem{lemma}{Lemma}

  \makeatother
\newenvironment{proof}[1][Proof]{\textbf{#1.} }{\ \rule{0.5em}{0.5em}}
\begin{document}

\author{José Manuel Corcuera\thanks{Universitat de Barcelona, Gran Via de les Corts
Catalanes, 585, E-08007 Barcelona, Spain. \texttt{E-mail: } jmcorcuera@ub.edu.
The work of J. M. Corcuera is supported by the Spanish grant MTM2013-40782-P.}.}
\title{The Golden Age of the Mathematical Finance}
\maketitle

\begin{abstract}
\noindent This paper is devoted to show that the last quarter of the past
century can be considered as the golden age of the Mathematical Finance. In
this period the collaboration of great economists and the best generation of
probabilists, most of them from the Strasbourg's School led by Paul André
Meyer, gave rise to the foundations of this discipline. They established the
two fundamentals theorems of arbitrage theory, close formulas for options, the
main modelling approaches and created the appropriate framework for the
posterior development.


\noindent\textbf{JEL-Classification}{\ C61$\cdot$ D43$\cdot$ D44$\cdot$
D53$\cdot$ G11$\cdot$ G12$\cdot$ G14}

\noindent\textbf{MS-Classification 2020}: 60G35, 62M20, 91B50, 93E03

\end{abstract}

\section{Introduction}

In the last quarter of the past century a new mathematical discipline emerged:
the Mathematical Finance. It was the conjunction of great economists like F.
Black, M. Scholes, R. Merton, at the United States, and many others, with
great mathematicians, most of them belonging to, or following, the famous
\emph{Séminaire de Probabilités de Strasbourg} under the leadership of Paul
Andrè Meyer. The theory of stochastic integrals with respect to
semimartingales developed in this Seminar was the mathematical basis to
establish the main results of the new discipline. In this paper we are going
to explain how arose and were proved the two main theorems and how they
provided the framework to the Arbitrage Theory, the core of the Mathematical
Finance. The number of results and new paths opened during this period was too
huge to be described in one single paper. Issues like local and stochastic
volatility models, \citeA{dupire94}, \citeA{heston93}, the different
approaches in credit risk, \citeA{dufsing99}, the use of strict local
martingales to describe bubbles, \citeA{lowigr00}, the Heath-Jarrow-Morton
approach in interest rate models, \citeA{hejamo92}, the models under
transaction costs,\citeA{leland85}, and a long etc, are not treated here.
However many excellent books on the matter, written before the end of the last
century, already include them, for instance, \citeA{lamlap96},
\citeA{musruk97}, \citeA{bjork98}, \citeA{danjea98} or \citeA{shiryaev99},
among others.

In the next section we explain the beginning, at the 70's, and the precedents.
In the third section we explain the two fundamental theorems and their proofs,
showing this intertwining among economists and mathematicians. Finally we
explain the arbitrage theory for pricing and some of the strategies to solve
the problem of incompleteness and the multiplicity of (non) arbitrage prices.

\section{The inception}

Probably \emph{everything} started at 1973 when Black and Scholes published
their famous paper, \emph{The pricing of options and corporate liabilities},
see \citeA{blasch73}, where they give a formula for a call option based on an
arbitrage argument and the Itô formula. They assume that the stock $S_{t}$
follows a geometric Brownian motion (a model proposed previously by
\citeA{samuelson65} ), that is
\[
\mathrm{d}S_{t}=S_{t}\left(  \mu\mathrm{d}t+\sigma\mathrm{d}W_{t}\right)
,S_{0}>0.
\]
with $B$ a Brownian motion, and that the value of the unit of money, say
$B_{t}$, evolves as
\[
\mathrm{d}B_{t}=B_{t}r\mathrm{d}t.
\]
for $\mu,\sigma,r$ constants. The way that Black and Scholes obtained the
formula is the following. Suppose that the price, at time $t$, of the call
option, where the payoff is $\left(  S_{T}-K\right)  ^{+}$, is a smooth
function of the form
\[
C_{t}:=f(t,S_{t}),
\]
and consider a portfolio with $\beta_{t}$ calls\ and $\alpha_{t}$ stocks, the
cost of this portfolio is
\[
\beta_{t}C_{t}+\alpha_{t}S_{t}=:V_{t},
\]
and when it evolves, in a self-financed way its value changes as
\[
\mathrm{d}V_{t}=\beta_{t}\mathrm{d}C_{t}+\alpha_{t}\mathrm{d}S_{t},
\]%
\[
\mathrm{d}V_{t}=\beta_{t}\left(  \partial_{t}f\mathrm{d}t+\partial
_{x}f\mathrm{d}S_{t}+\frac{1}{2}\partial_{xx}fS_{t}^{2}\sigma^{2}%
\mathrm{d}t\right)  +\alpha_{t}\mathrm{d}S_{t},
\]
where we apply the Itô formula for stochastic differentials.

Now if we take $\alpha_{t}=-$ $\beta_{t}\partial_{x}f$ we have that the cost
of this portfolio is
\[
\mathrm{d}V_{t}=\beta_{t}\left(  \partial_{t}f+\frac{1}{2}\partial_{xx}%
fS_{t}^{2}\sigma^{2}\right)  \mathrm{d}t.
\]
It behaves like a bank account! remember that if we put $V_{t}$ in the bank
account then $\mathrm{d}V_{t}=V_{t}r\mathrm{d}t,$ then if we assume an
equilibrium situation in such a way that it is not possible to do profit
without risk, we must have
\[
\beta_{t}\left(  \partial_{t}f+\frac{1}{2}\partial_{xx}fS_{t}^{2}\sigma
^{2}\right)  =rV_{t}=r(\beta_{t}C_{t}-\beta_{t}\partial_{x}fS_{t})=r\beta
_{t}(f-\partial_{x}fS_{t}).
\]
So, the price of a call\ is the solution of the partial differential equation%
\begin{equation}
\partial_{t}f+rx\partial_{x}f+\frac{1}{2}\sigma^{2}x^{2}\partial
_{xx}f=rf,\label{eq1}%
\end{equation}
with the boundary condition $f(T,x)=(x-K)_{+}$. By doing a change of variable
it is easy to obtain
\[
C_{t}=S_{t}\Phi(d_{+})-Ke^{-r(T-t)}\Phi(d_{-})
\]
where $\Phi$ is the c.d.f. of a standard normal distribution and
\[
d_{\pm}=\frac{\log(\frac{S_{t}}{K})+(r\pm\frac{1}{2}\sigma^{2})(T-t)}%
{\sigma\surd(T-t)}.
\]
If we apply the Feynman-Kac formula, it is easy to see that
\[
C_{t}=\mathbb{E}_{\ast}\left(  \left.  (S_{T}-K)_{+}e^{-r(T-t)}\right\vert
S_{t}\right)
\]
where $\mathbb{E}_{\ast}$ is an expectation assuming that
\[
\mathrm{d}S_{t}=rS_{t}\mathrm{d}t+\sigma S_{t}\mathrm{d}W_{t}^{\ast},
\]
where $W^{\ast}$ is a Brownian motion as well. In fact by the integration by
parts formula, the boundary condition and equation (\ref{eq1})%
\begin{align*}
(S_{T}-K)_{+}e^{-r(T-t)} &  =f(T,S_{T})e^{-r(T-t)}=f(t,S_{t})+\int_{t}%
^{T}e^{-r(u-t)}\partial_{x}f\sigma S_{u}\mathrm{d}W_{u}^{\ast}\\
&  +\int_{t}^{T}e^{-r(u-t)}\left(  \partial_{u}f+rS_{u}\partial_{x}f+\frac
{1}{2}\sigma^{2}S_{u}^{2}\partial_{xx}f\right)  \mathrm{d}u\\
&  -\int_{t}^{T}e^{-r(u-t)}rf\mathrm{d}u\\
&  =f(t,S_{t})+\int_{t}^{T}e^{-r(u-t)}\partial_{x}f\sigma S_{u}\mathrm{d}%
W_{u}^{\ast}.
\end{align*}
Now taking the conditional expectation we obtain the result.

However Samuelson and Merton in 1969, see also \citeA{merton73} where the
author extends the analytical method of Black and Scholes to price more
complex options, wrote a paper where they give the following formula for the
price of a call option with strike $K$
\begin{equation}
f(t,S_{t})=e^{-r(T-t)}\int_{\frac{S_{t}}{K}}^{\infty}\left(  zS_{t}-K\right)
\mathrm{d}Q_{T-t}\left(  z\right)  \label{SAMMER}%
\end{equation}
where
\[
\int_{0}^{\infty}z\mathrm{d}Q_{T-t}\left(  z\right)  =e^{r(T-t)}.
\]
Then, in the special case when $Q_{t}$ is a log-normal distribution with
log-variance equal to $\sigma^{2}(T-t)$ we recover the Black-Scholes formula.

Nevertheless we have to move back until 1900 to see the real birth of the
mathematical finance. At that time the French postgraduate student Louis
Bachelier presented his thesis \emph{Théorie de la Speculation}. As it is said
in \citeA{cokabrcrlele00} this thesis, supervised by Henri Poincaré, contains
ideas of enormous value in both finance and probability. Bachelier was the
first to consider the stochastic processes that today we call it Brownian or
Wiener process. He use it to model the movements of stock prices. He consider
two constructions of the process, as a limit of (today named) a random walk
and as a solution of a Fourier equation (today heat equation). Even he finds
the distribution of the supremum of the Brownian motion. From a financial
point of view he introduces a new idea to price an option. He uses the
principle of the mathematical expectation in games: a game is fair if the
expectation of the profit is zero, but he calculate the price by taking the
expectation with respect to what today we call a \emph{risk-neutral
probability}. He establishes, as a fundamental principle, that the
mathematical expectation of the profit of a buyer or seller of a financial
product has to be zero. This principle allows him to calculate the price of a
call option and he obtains, translated in the usual nowadays notation, see
\citeA[p. 50]{Bachelier1900}, that%

\begin{equation}
C_{0}=\int_{S_{0}-K}^{\infty}\left(  z+S_{0}-K\right)  \mathrm{d}Q_{T}\left(
z\right)  \label{Bach}%
\end{equation}
where $Q_{T}\sim N(0,\sigma^{2}T).$ It is important to remark that for
Bachelier the Gaussian distribution was the objective one, the historical one.
Even the same seems to be true for the Black and Scholes model: they probably
believed that the geometric Brownian motion was the real distribution of the
stock and there were not distinction between the real and the risk-neutral
probability. At least this  is how Merton interprets the Black-Scholes
formula. He says, in Merton 73, page 161: "...However, $\mathrm{d}Q$ [in
(\ref{SAMMER})] is a risk-adjusted ("util-prob") distribution, depending on
both, risk-preferences and aggregate supplies, while the distribution in
(\ref{eq1}) is the \emph{objective} distribution of returns on the common
stock". Moreover, note that (\ref{Bach}) is an \textit{arithmetic} version of
(\ref{SAMMER}), except for the discount factor $e^{-r(T-t)}$. It seems that,
at least in the French market, all payments were related to a single date, in
this way the discount factor is not needed in (\ref{Bach}), that is Bachelier
was talking about \textit{forward values}.

The way that Bachelor's thesis was appreciated at his time is a bit
controversial. It is an extended legend that Henri Poincaré, his supervisor,
undervalued it because according to \citeA{Bernstein93} he says
\textquotedblleft The topic is somewhat remote from those our candidates are
in the habit of treating\textquotedblright. But this appreciation has been
discussed in \citeA{cokabrcrlele00} where they present the Poincare's report
on the Bachelier thesis showing that Poincaré supported quite strongly the
dissertation. Also they show that the reason why Bachelier did not get a
position in Paris is not related with any bad opinion of Poincaré but because
Bachelier enrolled in the army at the First World War. Moreover
\citeauthor{Bernstein93} tells a story about how the pioneers in mathematical
finance became aware about Bachelor's thesis suggesting that his work was
ignored for a long period, but there is evidence that Bachelor's work was more
known than the above legend says, Kolmogorov, Keynes, Feller, Lévy, among
others, knew it.

Following the works of Black, Scholes and Merton, Harrison and Kreps,
\citeA{harkre79} realize that the absence of arbitrage can be characterized by
the existence of a risk-neutral probability and that, adding a notion of
completeness, price of derivatives could be obtained as expectations of
discounted payoffs with respect to the risk-neutral probability, they analyze
the discrete time case and where the probability space is finite. Later
Harrison and Pliska in \citeA{harplis81} consider the continuous time case and
where the stock process is a semimartingale. They show that the existence of a
local martingale measure (or risk-neutral measure) implies absence of
arbitrage. They also define what is a complete market and establish a relation
with the theorems about martingale representation and the set of risk-neutral
probabilities. They prove that if the market is complete the set of
risk-neutral probabilities is a singleton. This latter issue is treated in
more detail in their paper \citeA{harplis93}. These papers established what
are known as the two fundamental theorems of asset pricing, the first one is
the characterization of the arbitrage and the second that of completeness.
These theorems open the way for extensions of the Black-Scholes model and new
price formulas of new derivatives. Curiously the semimartingale theory, the
only applicable since the semimartingale set of process is invariant under
changes of equivalent probabilities, was being developed in these decades but
without a direct relation with the issue. Harrison and Pliska say: "the parts
of probability theory most relevant to the general question (about what
processes for the stock yield a complete market) are those results, usually
abstract in appearance and French in origin...". They are refereeing to the
Strasbourg Seminaire led by André Meyer. They also add "...We have started to
feel that all the standard problems studied in martingale theory and all the
major results must have interpretations and applications in our setting".

\section{The fundamental theorems of asset pricing}

The term Fundamental Theorems of Asset Pricing (FTAP) was coined by Phil
Dybvig in 1987, to describe the work initiated by his thesis advisor, Steve
Ross around 1978, see \citeA{ross78}. The connection to martingale theory is
in \citeA{harkre79} with an important extension to the continuous time setting
in \citeA{harplis81}. In fact there are two fundamental theorems as we will see.

Let $S=$ $\left(  S_{t}^{0},S_{t}^{1}...,S_{t}^{d}\right)  _{t\in I}$ be a
non-negative $d+1$-dimensional semimartingale representing the price process
of $d+1$ securities, here $I$ is either a discrete finite set $I=\{0,1,..,T\}$
or a compact interval $I=[0,T].$ This process is assumed to be defined in a
complete filtered probability space $\left(  \Omega,\mathcal{F},\mathbb{F}%
,\mathbb{P}\right)  $ where $\mathbb{F=}\left(  \mathcal{F}_{t}\right)  _{t\in
I}$ is a filtration \ representing the flow of information and satisfying the
\emph{usual conditions}. We suppose that $S_{t}^{0}>0$ is $\mathcal{F}_{t-1}%
$-measurable, $t=1,...,T$, that is $S^{0}$ is predictable. We also assume that
$\mathcal{F}_{0}$ is trivial (a.s.) and that $\mathcal{F}_{T}=\mathcal{F}$.
\ We define the \emph{discounted price }as\emph{ }%

\[
\tilde{S}_{t}:=\frac{S_{t}}{S_{t}^{0}},\text{ }t\in I.
\]

\subsection{Discrete time setting}

\begin{definition}
A \emph{trading strategy} is a predictable stochastic process $\ \phi
=((\phi_{t}^{0},\phi_{t}^{1},...,\phi_{t}^{d}))_{1\leq t\leq T}$\ in $R^{d+1}%
$, that is $\phi_{t}^{i}$ is $\mathcal{F}_{t-1}$-measurable, \ for all $1\leq
t\leq T$.
\end{definition}

$\phi_{t}^{i}$\ indicates the number of units of $\ $security $i$ in the
portfolio at time $t$ and that $\phi$ is \ \textit{predictable} means that the
positions in the portfolio at $t$ is decided at $t-1,$ using the information
available in $\mathcal{F}_{t-1}$. Then we have also the following definitions.

\begin{definition}
The value of the portfolio associated with a trading strategy $\phi$ is given
by
\[
V_{t}(\phi)=\phi_{t}\cdot S_{t}:=\sum_{i=0}^{d}\phi_{t}^{i}S_{t}^{i},\quad
t=1,...,T,\quad V_{0}(\phi)=\phi_{1}\cdot S_{0}.
\]
and its \textit{discounted} value
\[
\tilde{V}_{t}(\phi):=\phi_{t}\cdot\tilde{S}_{t}%
\]

\end{definition}

\begin{definition}
A trading strategy $\phi$ is said to be \emph{self-financing}\textit{ }if
\[
V_{t}(\phi)=\phi_{t+1}\cdot S_{n},\quad t=1,...,T-1.
\]

\end{definition}

It is easy to see that the \emph{self-financing} condition is equivalent to
the equality%
\[
\tilde{V}_{t}(\phi)=\tilde{V}_{0}(\phi)+\sum_{i=1}^{t}\phi_{t}\cdot
\Delta\tilde{S}_{t},\text{ }t=1,...,T.
\]

\begin{definition}
An \emph{admissible} trading strategy $\phi$ is a \emph{self-financing}
strategy satisfying the following constraint%
\[
V_{t}\left(  \phi\right)  \geq0\text{ a.s. for }t=0,1,...,T
\]

\end{definition}

\begin{definition}
\emph{An admissible trading strategy }$\phi$ is an \emph{arbitrage opportunity
if it satisfies }%
\begin{align*}
V_{0}\left(  \phi\right)   &  =0,\text{ a.s. }\\
V_{T}\left(  \phi\right)   &  \geq0,\text{ a.s. and }\mathbb{P}\left(
V_{T}>0\right)  >0.
\end{align*}

\end{definition}

In an analogous way we can say when a self-financing strategy is an arbitrage.
In the discrete time setting we have the following lemma

\begin{lemma}
The class of admissible trading strategies contains no arbitrage opportunities
if and only if the class of self-financing strategies contains no arbitrage opportunities.
\end{lemma}

\begin{proof}
Let $\varphi$ \ be a self-financing strategy that is an arbitrage. Define
\[
t=\inf\{u,\tilde{V}_{t}(\varphi)\geq0\text{ a.s. \ for all }u>t\text{ }\},
\]
note that $t\leq T-1$\ since $\tilde{V}_{T}(\varphi)\geq0.$\ Let $A=\left\{
\tilde{V}_{t}(\varphi)<0\right\}  $,\ define the predictable vector process
$\theta$, such that for all $i=1,...,d$\
\[
\theta_{u}^{i}=\left\{
\begin{tabular}
[c]{ll}%
$0$ & $u\leq t$\\
$\mathbf{1}_{A}\varphi_{u}^{i}$ & $u>t$%
\end{tabular}
\ \ \ \ \ \ \ \ \ \ \ \ \ \right.
\]
Then, $\tilde{V}_{u}()=0,$ for all $0\leq u\leq t$ and for all $u>t$\
\begin{align*}
\tilde{V}_{u}(\theta)  &  =\sum_{v=t+1}^{u}\mathbf{1}_{A}\varphi_{v}%
\cdot\Delta\tilde{S}_{v}=\mathbf{1}_{A}\left(  \sum_{v=1}^{u}\varphi_{v}%
\cdot\Delta\tilde{S}_{v}-\sum_{v=1}^{t}\varphi_{v}\cdot\Delta\tilde{S}%
_{v}\right) \\
&  =\mathbf{1}_{A}\left(  \tilde{V}_{u}(\varphi)-\tilde{V}_{t}(\varphi
)\right)  \geq0,
\end{align*}
so $\theta$ is admissible and $\tilde{V}_{T}(\theta)=\mathbf{1}_{A}\left(
\tilde{V}_{T}(\varphi)-\tilde{V}_{t}(\varphi)\right)  >0$ in $A,$ then
$\theta$ is an admissible arbitrage.
\end{proof}

Notice that, according to this lemma, if we consider only the admissible
strategies we are not reducing the possibilities of arbitrage opportunities.
Henceforth the term \emph{arbitrage opportunities} refers to \emph{admissible}
arbitrage opportunities.

\begin{definition}
Probabilities $\mathbb{P}$ and $\mathbb{P}^{\ast}$ defined on $\left(
\Omega,\mathcal{F}\right)  $ are said to be equivalent, we write
$\mathbb{P}\sim\mathbb{P}^{\ast}$, if they have the same null-sets.
\end{definition}

Now we can establish the first fundamental theorem of the asset pricing (FFTAP).

\begin{theorem}
There are not arbitrage opportunities if and only if there is a probability
$\mathbb{P}^{\ast}\sim$ $\mathbb{P}$ under which the process $\tilde{S}$ is a martingale.
\end{theorem}

\begin{proof}
(Sufficiency) Suppose that $\tilde{S}$ is a martingale under some
$\mathbb{P}^{\ast}\sim$ $\mathbb{P}$ . Then, let $\phi$ be an arbitrage
opportunity, so, since it is admissible, $V_{t}\left(  \phi\right)  \geq0$
a.s. $\mathbb{P}^{\ast}$ (because the equivalence between $\mathbb{P}^{\ast}$
and $\mathbb{P}$) and $\mathbb{E}_{\mathbb{P}^{\ast}}\left(  V_{t}\left(
\phi\right)  \right)  $ is well defined, although possibly infinite. Then we
have
\begin{align*}
\mathbb{E}_{\mathbb{P}^{\ast}}\left(  \tilde{V}_{t}\left(  \phi\right)
\right)   &  =\mathbb{E}_{\mathbb{P}^{\ast}}\left(  \phi_{t}\cdot\tilde{S}%
_{t}\right)  =\mathbb{E}_{\mathbb{P}^{\ast}}\left(  \mathbb{E}_{\mathbb{P}%
^{\ast}}\left(  \left.  \phi_{t}\cdot\tilde{S}_{t}\right\vert \mathcal{F}%
_{t-1}\right)  \right)  \\
&  =\mathbb{E}_{\mathbb{P}^{\ast}}\left(  \phi_{t}\cdot\mathbb{E}%
_{\mathbb{P}^{\ast}}\left(  \left.  \tilde{S}_{t}\right\vert \mathcal{F}%
_{t-1}\right)  \right)  =\mathbb{E}_{\mathbb{P}^{\ast}}\left(  \phi_{t}%
\cdot\mathbb{E}_{\mathbb{P}^{\ast}}\left(  \left.  \tilde{S}_{t}\right\vert
\mathcal{F}_{t-1}\right)  \right)  \\
&  =\mathbb{E}_{\mathbb{P}^{\ast}}\left(  \phi_{t}\cdot\tilde{S}_{t-1}\right)
=\mathbb{E}_{\mathbb{P}^{\ast}}\left(  \phi_{t-1}\cdot\tilde{S}_{t-1}\right)
=...=\mathbb{E}_{\mathbb{P}^{\ast}}\left(  \tilde{V}_{0}\left(  \phi\right)
\right)  =0.
\end{align*}
where in the second line line we use the predictability of $\phi$ and the
martingale property of $\tilde{S}$,\ and in the third line the self-financing condition.

(Necessity)\ Now, if the number of elements in $\Omega$ is finite we have
simpler proofs than in the general case. For the finite case the standard
proof is that of \citeA{harplis81}, based on the separation hyperplane theorem
in $\mathbb{R}^{\left\vert \Omega\right\vert }$. \ If the sample space is not
finite \citeA{morton88} and \citeA{damowi90} give a proof based in the
following lemma for the case that $d=1$.

\begin{lemma}
Let $Y$ be a bounded $d$-dimensional random vector defined on some $\left(
\Omega,\mathcal{F},\mathbb{P}\right)  $ \ and assume that $Y\in K$ \ for some
compact set in $\mathbb{R}^{d}$ then one of the following two conditions
hold,
\begin{align*}
&  \left.  1.\text{ There exists }\alpha\in\mathbb{R}^{d}\text{ with }%
\alpha\cdot Y\geq0\text{ a.s and }\mathbb{P(}\alpha\cdot Y>0)>0.\right. \\
&  \left.  2.\text{ There exists a positive }g\in C(K)\text{ (the space of
real continuous functions on }K\text{) with }\right. \\
&  \left.  \text{
\ \ \ \ \ \ \ \ \ \ \ \ \ \ \ \ \ \ \ \ \ \ \ \ \ \ \ \ \ \ \ \ \ \ \ \ \ \ \ \ \ \ \ \ \ \ \ \ \ \ \ \ \ \ \ \ }%
\mathbb{E(}g\mathbb{(}Y\mathbb{)}Y)=0.\right.
\end{align*}

\end{lemma}
\end{proof}

For $d>1$ the idea is patching together the conditional martingale measures
corresponding to each period but this involves subtle arguments by using the
Measurable Selection Theorem. For alternative proofs that try to simplify the
above one, see \citeA{schachermayer92}, \citeA{rogers94} and \citeA{jacshi98}.

\bigskip Before stating the second fundamental theorem of asset pricing
(SFTAP) we need the following definitions.

\begin{definition}
Let $X$ be a claim, that is $X\geq0$, and $X\in\mathcal{F}_{T}$. We say that
$X$ is attainable, or replicable, if $X$ is equal to the final value of and
admissible strategy.
\end{definition}

Assume that the market model is free of arbitrage, then the set of
risk-neutral probabilities is not empty. If we choose one, say $\mathbb{P}%
^{\ast}$, we have the following definition.

\begin{definition}
We say that the market model is \emph{complete} if every claim $X\in
\mathcal{F}_{T},$ such that $\mathbb{E}_{\mathbb{P}^{\ast}}\left(  \tilde
{X}\right)  <\infty$, is attainable.
\end{definition}

Then the SFTAP reads as follows.

\begin{theorem}
The market is complete if and only if the risk-neutral probability is unique
\end{theorem}

\begin{proof}
(The 'only if' part) Assume that the market is complete, then, the claims of
the form $S_{T}^{0}\mathbf{1}_{A}$ with $A\in\mathcal{F}_{T}$ are attainable,
therefore we can write%
\[
\mathbf{1}_{A}=\tilde{V}_{0}(\phi)+\sum_{t=1}^{T}\phi_{t}\cdot\Delta\tilde
{S}_{t}%
\]
where $\phi$ is admissible. Now, if we assume that there are two risk-neutral
probabilities, say $\mathbb{P}^{\ast}$ and $\mathbb{Q}^{\ast}$, we have, since
$\tilde{S}$ is a martingale with respect to both probabilities $\mathbb{P}%
^{\ast}$ and $\mathbb{Q}^{\ast}$, that
\[
\mathbb{E}_{\mathbb{P}^{\ast}}\left(  \mathbf{1}_{A}\right)  =\mathbb{E}%
_{\mathbb{Q}^{\ast}}\left(  \mathbf{1}_{A}\right)  ,
\]
in such a way that $\mathbb{P}^{\ast}$and $\mathbb{Q}^{\ast}$ are the same
probability in $\mathcal{F}_{T}$.

(The 'if' part ) Assume first that $\Omega$ is finite. Let $H$ be the subset
of random variables of the form
\[
\tilde{V}_{0}\left(  \phi\right)  +\sum_{t=1}^{T}\phi_{t}\cdot\Delta\tilde
{S}_{t}%
\]
with $\phi$ predictable. $H$\ is a vector subspace of the vectorial space, say
$E$, formed by all random variables. Moreover it is not a trivial subspace, in
fact since the market is incomplete there will exist $h$ such that $\frac
{h}{S_{n}^{0}}\not \in H$ \ (note that if $h\geq0$ can be replicated by a
non-admissible strategy then the market cannot be viable). Let $\mathbb{P}%
^{\ast}$ be a risk-neutral probability in $E$, we can define the scalar
product $\langle X,Y\rangle=\mathbb{E}_{\mathbb{P}^{\ast}}(XY)$. Let $X$ be an
random variable orthogonal to $H$ and set
\[
\mathbb{P}^{\ast\ast}(\omega)=\left(  1+\frac{X(\omega)}{2||X||_{\infty}%
}\right)  \mathbb{P}^{\ast}(\omega).
\]
Then we have an equivalent probability to $\mathbb{P}^{\ast}:$%
\[
\mathbb{P}^{\ast\ast}(\omega)=\left(  1+\frac{X(\omega)}{2||X||_{\infty}%
}\right)  \mathbb{P}^{\ast}(\omega)>0,
\]

\
\[
\sum\mathbb{P}^{\ast\ast}(\omega)=\sum\mathbb{P}^{\ast}(\omega)+\frac
{\mathbb{E}_{\mathbb{P}^{\ast}}(X)}{2||X||_{\infty}}=\sum\mathbb{P}^{\ast
}(\omega)=1.
\]
Now, take $\mathbf{1}_{A},$ with $A\in\mathcal{F}_{t-1}.$ Then $\mathbf{1}%
_{A}\Delta\tilde{S}_{t}^{j}\in H$\ $\ $for all $j=1,..,d$ and $X$ is
orthogonal to $H,$ therefore \
\[
\mathbb{E}_{\mathbb{P}^{\ast\ast}}\left(  \mathbf{1}_{A}\Delta\tilde{S}%
_{t}^{j}\right)  =\mathbb{E}_{\mathbb{P}^{\ast}}\left(  \mathbf{1}_{A}%
\Delta\tilde{S}_{t}^{j}\right)  +\frac{\mathbb{E}_{\mathbb{P}^{\ast}}\left(
X\mathbf{1}_{A}\Delta\tilde{S}_{t}^{j}\right)  }{2||X||_{\infty}}=0,
\]
in such a way that $\tilde{S}$\ is a $\mathbb{P}^{\ast\ast}$-martingale. The
surprise for the general case is that if the risk-neutral probability is
unique then $\Omega$ is essentially finite for $\mathbb{P}$!, in the sense
that $\mathcal{F}_{T}$ is purely atomic with respect to $\mathbb{P}$ with at
most $(d+1)^{N}$ atoms, see Theorem 6 in \citeA{jacshi98}.
\end{proof}

\subsection{Continuous time setting}

Now we have that $I=[0,T],$ $S=$ $\left(  S_{t}^{0},S_{t}^{1}...,S_{t}%
^{d}\right)  _{t\in I}$ is a non-negative $d+1$-dimensional semimartingale
representing the price process of $d+1$ securities. This process is assumed to
be defined in a complete filtered probability space $\left(  \Omega
,\mathcal{F},\mathbb{F},\mathbb{P}\right)  $ where $\mathbb{F=}\left(
\mathcal{F}_{t}\right)  _{t\in I}$ is a filtration \ satisfying the
\emph{usual conditions}. We suppose that $S_{t}^{0}>0$. We also assume that
$\mathcal{F}_{0}$ is trivial (a.s.) and that $\mathcal{F}_{T}=\mathcal{F}$.
\ We define the \emph{discounted value }as\emph{ }%

\[
\tilde{S}_{t}:=\frac{S_{t}}{S_{t}^{0}},\text{ }t\in I.
\]

\begin{definition}
A trading strategy is a predictable stochastic process\ $\phi=((\phi_{t}%
^{0},\phi_{t}^{1},...,\phi_{t}^{d}))_{1\leq t\leq T}$ in $\mathbb{R}^{d+1}$,
that means that $\phi_{t}^{i}$ is measurable with respect to the sigma-field
generated by the càglàd adapted processes and that is integrable with respect
to $S$.
\end{definition}

Then we have also the following definitions.

\begin{definition}
The discounted value of the portfolio associated with a trading strategy
$\phi$ is given by
\[
\tilde{V}_{t}(\phi)=\phi_{t}\cdot\tilde{S}_{t}=\sum_{i=0}^{d}\phi_{t}%
^{i}\tilde{S}_{t}^{i},
\]

\end{definition}

\begin{definition}
A trading strategy $\phi$ is said to be \emph{self-financing}\textit{ }if
\[
\tilde{V}_{t}(\phi)=\phi_{0}\cdot\tilde{S}_{0}+\int_{0}^{t}\phi_{t}%
\cdot\mathrm{d}\tilde{S}_{t},\quad t\in\lbrack0,T].
\]

\end{definition}

\begin{definition}
An \emph{admissible} trading strategy $\phi$ is a \emph{self-financing}
strategy satisfying%
\[
\tilde{V}_{t}\left(  \phi\right)  \geq-a\text{ . for all }t\in\lbrack0,T]
\]
for some $a>0$.
\end{definition}

According to Delbaen and Schachermayer (1994) we use the notation%
\begin{align*}
&  \left.  K_{0}=\left\{  \int_{0}^{T}\phi_{s}\cdot\mathrm{d}\tilde{S}%
_{t},\phi\text{ admissible}\right\}  \right. \\
&  \left.  C_{0}=K_{0}-L_{+}^{0}\right. \\
&  \left.  K=K_{0}\cap L^{\infty}\right. \\
&  \left.  C=C_{0}\cap L^{\infty}\right. \\
&  \left.  \overline{C}\text{ the closure of }C\text{ under }L^{\infty
}\right.
\end{align*}

\begin{definition}
We say that the model satisfies the No Free Lunch with Vanishing Risk
condition (NFLVR) if $\overline{C}\cap L_{+}^{\infty}=\{0\}$
\end{definition}

We have the FFTAP in continuous time

\begin{theorem}
Let $\tilde{S}$ be a locally bounded $\mathbb{R}^{d}$-valued semimartingale.
There are not free lunches with vanishing risk if and only if there is
probability $\mathbb{P}^{\ast}\sim\mathbb{P}$ under which $\tilde{S}$ is a
local martingale.
\end{theorem}

\begin{proof}
(Sufficiency) Let $\phi$ be an admissible strategy with initial value equal to
zero and let $\mathbb{P}^{\ast}$ be a probability such that $\tilde{S}$ \ is a
$\mathbb{P}^{\ast}$-local martingale. Since $\tilde{V}_{\cdot}(\phi)=\int%
_{0}^{\cdot}\phi_{t}\cdot\mathrm{d}\tilde{S}_{t}$ is bounded below it is a
local martingale (see \citeA{ansstr94}) and in fact a supermartingale, by
Fatou's lemma. Then we have
\[
\mathbb{E}_{\mathbb{P}^{\ast}}\left(  \tilde{V}_{T}(\phi)\right)  \leq0
\]
and we obtain that $\mathbb{E}_{\mathbb{P}^{\ast}}\left(  h\right)  \leq0$ for
every $h$ in $C$ and consequently for $h$ in $\bar{C}$ by the Lebesgue
theorem. Therefore $\bar{C}\cap L_{+}^{\infty}=\{0\}.$

(Necessity) This is the difficult part, but essentially it consists in
considering the weak* topology in $L^{\infty}$ and to show that, under the
NFLVR condition, $C$\ is closed with this topology (the proof of this is very
technical and the reference is \citeA{delsch94}). Then we can apply the
separation theorem \emph{with this weak topology} (e.g., 9.2 in
\citeA{schaefer71}) and to show that there exist $\mathbb{P}^{\ast}%
\sim\mathbb{P}$ such that $\mathbb{E}_{\mathbb{P}^{\ast}}\left(  h\right)
\leq0$ for each $h$ in $C$ (see the details in \citeA{schachermayer94}, 3.1).
Now if we assume first that $\tilde{S}$ is a bounded semimartingale we have
that for each $s<t,$ $B\in\mathcal{F}_{s}$ $\ $and $\alpha\in\mathbb{R}$,
$\alpha\mathbf{1}_{B}\left(  \tilde{S}_{t}-\tilde{S}_{s}\right)  \in C$ .
Therefore $\mathbb{E}_{\mathbb{P}^{\ast}}\left(  \mathbf{1}_{B}\left(
\tilde{S}_{t}-\tilde{S}_{s}\right)  \right)  =0$ and $\mathbb{P}^{\ast}$ is a
risk-neutral probability for $\tilde{S}$. If $\tilde{S}$ is locally bounded
then a localization argument, and the result for the bounded case, allows us
to obtain a locally martingale measure $\mathbb{P}^{\ast}$ for $\tilde{S}.$
\end{proof}

Also we have a result for the case that $\tilde{S}$ $\ $is not locally
bounded. The result involves the concept of $\sigma$-martingale, that is a
process that can be obtained as a stochastic integral of a \emph{positive}
integrand with respect to a martingale, see \citeA{delsch99}. The result is
the following:

\begin{theorem}
Let $\tilde{S}$ be a $\mathbb{R}^{d}$-valued semimartingale. There are not
free lunches with vanishing risk if and only if there is probability
$\mathbb{P}^{\ast}\sim\mathbb{P}$ under which $\tilde{S}$ is a $\sigma$-martingale.
\end{theorem}

We have the analogous definitions for completeness in the the continuous time
case, thought the meaning of replication now is in terms of stochastic integrals.

\begin{definition}
Let $X$ be a claim, that is $X\geq0$, and $X\in\mathcal{F}_{T}$. We say that
$X$ is attainable, or replicable, if $X$ is equal to the final value of and
admissible strategy.
\end{definition}

Assume that the market model satisfies the NFLVR condition, then the set of
risk-neutral probabilities is not empty. If we choose one, say $\mathbb{P}%
^{\ast}$, we have the following definition.

\begin{definition}
We say that the market model is \emph{complete} if every claim $X\in
\mathcal{F}_{T},$ such that $\mathbb{E}_{\mathbb{P}^{\ast}}\left(  \tilde
{X}\right)  <\infty$, is attainable.
\end{definition}

Then the SFTAP for the continuous time setting reads the same way as in the
discrete time case.

\begin{theorem}
The market is complete if and only if the risk-neutral probability is unique
\end{theorem}

\begin{proof}
(The 'only if' part) The proof is exactly the same as in the discrete time case.

(The 'if' part) Suppose that the risk neutral probability is unique. By
Theorem 11.2 in \citeA{jacod79} $\tilde{S}$ has \emph{the representation
property with respect to }$\mathbb{P}^{\ast}$ (any $\mathbb{P}^{\ast}%
$-martingale can we written as an integral with respect to $\tilde{S}$) if and
only $\mathbb{P}^{\ast}$ is an \emph{extremal} measure in the set of all
probability measures for which $\tilde{S}$ is a local martingale. \ Suppose
that $\mathbb{P}^{\ast}$ is not extremal, then there exist $\mathbb{Q}$,
$\ \mathbb{Q}^{\prime}$ local martingale measures (not equivalent to
$\mathbb{P}$) such that $\mathbb{P}^{\ast}=\lambda\mathbb{Q+}(1-\lambda
)\mathbb{Q}^{\prime}$ for some $\lambda\in(0,1)$. \ Also $\mathbb{Q}_{\beta
}:=\beta\mathbb{Q+}(1-\beta)\mathbb{Q}^{\prime}$ is \ a local martingale
measure for any $\beta\in(0,1),$ and \ since $\mathbb{Q<<P}^{\ast}$ and
$\mathbb{Q}^{\prime}\mathbb{<<P}^{\ast}$ $\ $\ we obtain that $\mathbb{Q}%
_{\beta}\sim\mathbb{P}^{\ast}$ contradicting the uniqueness of $\ \mathbb{P}%
^{\ast}$. \ Now by the representation property, if $\ \mathbb{E}%
_{\mathbb{P}^{\ast}}\left(  \tilde{X}\right)  <\infty$ we can write%
\[
\mathbb{E}_{\mathbb{P}^{\ast}}\left(  \tilde{X}|\mathcal{F}_{t}\right)
=\mathbb{E}_{\mathbb{P}^{\ast}}\left(  \tilde{X}\right)  +\int_{0}^{t}\phi
_{s}\cdot\mathrm{d}\tilde{S}_{s},t\in\lbrack0,T]
\]
and, in particular, with $\phi_{0}\cdot\tilde{S}_{0}=\mathbb{E}_{\mathbb{P}%
^{\ast}}\left(  \tilde{X}\right)  $ ,%
\[
\tilde{X}=\phi_{0}\cdot\tilde{S}_{0}+\int_{0}^{T}\phi_{t}\cdot\mathrm{d}%
\tilde{S}_{t}.
\]
Notice that the admissibility condition for $\phi$ is trivially satisfied
since $X\geq0.$
\end{proof}

\section{The arbitrage theory}

From the the previous results we deduce that in a complete market model we can
price any payoff $X$ such that $\mathbb{E}_{\mathbb{P}^{\ast}}\left(
\tilde{X}\right)  <\infty$ and the price at time $t$ is obviously the price of
the replicating portfolio, $V_{t}\left(  \phi\right)  ,$ given by the strategy
$\phi$ and such that $\tilde{V}_{t}\left(  \phi\right)  $ is a $\mathbb{P}%
^{\ast}$ martingale. Then since $\tilde{X}=$ $\tilde{V}_{T}\left(
\phi\right)  $ we have that
\begin{equation}
V_{t}\left(  \phi\right)  =S_{t}^{0}\mathbb{E}_{\mathbb{P}^{\ast}}\left(
\left.  \tilde{X}\right\vert \mathcal{F}_{t}\right)  . \label{price}%
\end{equation}
Notice that we can have other admissible strategies, say $\varphi$,
replicating $X$ but such that, the corresponding discounted portfolios, are
local martingales and since they are bounded below they are supermartingales.
Then
\[
\tilde{V}_{t}\left(  \phi\right)  =\mathbb{E}_{\mathbb{P}^{\ast}}\left(
\left.  \tilde{X}\right\vert \mathcal{F}_{t}\right)  =\mathbb{E}%
_{\mathbb{P}^{\ast}}\left(  \left.  \tilde{V}_{T}\left(  \varphi\right)
\right\vert \mathcal{F}_{t}\right)  \leq\tilde{V}_{t}\left(  \varphi\right)
.
\]
Consequently they are more expensive and no one would pay for that.

If the model is not complete but satisfying NFLVR, the formula (\ref{price})
can be used for pricing if we take one risk-neutral probability since, if we
do that, we will have a price market model free of arbitrage. But what is the
correct risk-neutral probability? One way to select a risk-neutral probability
is to ask an additional property to the risk-neutral probability. An
alternative is to choose a risk-neutral probability, say $\mathbb{P}^{\ast}$,
\emph{close} to $\mathbb{P}$. In this sense if we take a strictly convex
function we can try to minimize
\begin{equation}
\mathbb{E}\left(  V\left(  \frac{\mathrm{d}\mathbb{P}^{\ast}}{\mathrm{d}%
\mathbb{P}}\right)  \right)  . \label{min}%
\end{equation}
For instance if $V(x)=x\log x$ we have \emph{the minimal entropy martingale}
measure, see \citeA{frittelli00}, that coincides with the Esscher measure in
the case we model the stock by a geometric Lévy model, see \citeA{chan99}.
Some risk-neutral measures, like \emph{the minimal martingale measure} are
also related to the minimization of the cost to replicate perfectly the
contingent claims, see \citeA{folsch91}. It can be seen that, in certain
cases, \emph{the minimal} \emph{martingale measure} minimizes (\ref{min}) with
$V(x)=\frac{x^{2}}{2},$ see \citeA{chan99} and \citeA{schweizer95}.

If we consider the (negative) Legendre transform of $V$ we have
\[
U(x)=\inf_{y}\left(  V(y)+xy\right)  ,
\]
for the examples above we obtain $U(x)=-e^{-x}$ if $V(x)=x\log x$ and
$U(x)=-\frac{x^{2}}{2}$ if $V(x)=\frac{x^{2}}{2}$, \ that can be interpreted
as \emph{utility }functions. Then by the duality relationship we have that,
under appropriate conditions see \cite{schachermayer00}, the optimal wealth,
say $W_{T}^{\ast}$ , when we try to maximize the expected utility of the final
wealth (by using admissible strategies), satisfies
\[
\frac{U^{\prime}(W_{T}^{\ast})}{\mathbb{E}\left(  U^{\prime}(W_{T}^{\ast
})\right)  }=\frac{\mathrm{d}\mathbb{P}^{\ast}}{\mathrm{d}\mathbb{P}}%
\]
where $\frac{\mathrm{d}\mathbb{P}^{\ast}}{\mathrm{d}\mathbb{P}}$ minimize
(\ref{min}). Then if we use this risk-neutral probability to price derivatives
what we obtain is \emph{the marginal utility} \emph{indifference price}
proposed by Mark H.A. Davis, see \citeA{davis97}. In fact, let define, for an
initial wealth $x$,
\[
v(x):=\sup\mathbb{E}\left(  U\left(  W_{T,x}\right)  \right)
\]
assuming $U:\mathbb{R}_{+}\rightarrow\mathbb{R}$ is strictly concave and
strictly increasing $\mathcal{C}^{2}$ function, with $U^{\prime}(\infty)=0$
\ and $U^{\prime}(0^{+})=\infty$. If we invest $\varepsilon$ in a claim $\xi$
with price $p$ we can look for the value of $p,$ say $\hat{p},$ such that
\[
\frac{\mathrm{d}}{\mathrm{d}\varepsilon}\sup\mathbb{E}\left(  U\left(
W_{T,x-\varepsilon}+\frac{\varepsilon}{\hat{p}}\xi\right)  \right)
_{\varepsilon=0}=0,
\]
then, we have that
\[
\mathbb{E}\left(  U\left(  W_{T,x-\varepsilon}+\frac{\varepsilon}{p}%
\xi\right)  \right)  -\mathbb{E}\left(  U\left(  W_{T,x-\varepsilon}\right)
\right)  =\varepsilon\mathbb{E}\left(  U^{^{\prime}}\left(  W_{T,x-\varepsilon
}\right)  \frac{\xi}{p}\right)  +o\left(  \varepsilon\right)
\]
and from here we obtain that
\[
\hat{p}=\frac{\mathbb{E}\left(  U^{^{\prime}}\left(  W_{T,x}^{\ast}\right)
\xi\right)  }{v^{\prime}(x)},
\]
since $v^{\prime}(x)=\mathbb{E}\left(  U^{^{\prime}}\left(  W_{T,x}^{\ast
}\right)  \right)  $ we obtain that
\[
\hat{p}=\mathbb{E}_{\mathbb{P}^{\ast}}\left(  \xi\right)  ,\text{ }%
\]
with $\frac{\mathrm{d}\mathbb{P}^{\ast}}{\mathrm{d}\mathbb{P}}=\frac
{U^{\prime}(W_{T,x}^{\ast})}{\mathbb{E}\left(  U^{\prime}(W_{T,x}^{\ast
})\right)  }.$

\bibliographystyle{apacite}
\bibliography{refhistor}

\end{document}